\documentclass[journal]{IEEEtran}
\IEEEoverridecommandlockouts
\usepackage{cite}
\usepackage{amsmath,amssymb,amsfonts}
\usepackage{graphicx}
\usepackage{textcomp}
\usepackage{xcolor}
\usepackage{amssymb}
\newcommand*{\Scale}[2][4]{\scalebox{#1}{$#2$}}%
\usepackage{enumitem}
\usepackage{booktabs}

\usepackage{nomencl}
\makenomenclature
\usepackage{array}

\makeatletter
\newcommand{\thickhline}{%
    \noalign {\ifnum 0=`}\fi \hrule height 1pt
    \futurelet \reserved@a \@xhline
}
\newcolumntype{"}{@{\hskip\tabcolsep\vrule width 1pt\hskip\tabcolsep}}
\makeatother

\usepackage[utf8]{inputenc}
\usepackage[english]{babel}
\usepackage{comment}
\usepackage{rotating}
\usepackage{blindtext}
\usepackage{CJK}
\usepackage{amsmath}
\usepackage{subcaption}

\usepackage{setspace}
\usepackage{caption}
\usepackage{graphicx}
\usepackage{tabularx}
\usepackage{mwe}
\usepackage{booktabs}
\usepackage{multirow}
\usepackage{amsthm}

\newtheorem{myCla}{Claim}

\usepackage{mathtools}
\usepackage{algpseudocode}
\usepackage{amsmath}
\usepackage{amssymb}
\usepackage{verbatim}
\usepackage{hyperref} 
\usepackage{bm}

\usepackage{textcomp}

\usepackage{lipsum}

\def\BibTeX{{\rm B\kern-.05em{\sc i\kern-.025em b}\kern-.08em
    T\kern-.1667em\lower.7ex\hbox{E}\kern-.125emX}}
\begin{document}

\title{Proactive rebalancing and speed-up techniques for on-demand high capacity ridesourcing services}
\author{Yang~Liu,
        Samitha~Samaranayake
\thanks{Y. Liu and S. Samaranayake are with the School of Civil and Environmental Engineering, Cornell University, Ithaca, NY, 14850, USA. Their e-mails are yl2464@cornell.edu and samitha@cornell.edu, respectively.}
}


\maketitle
\begin{abstract}
We present a probabilistic proactive rebalancing method and speed-up techniques for improving the performance of a state-of-the-art real-time high-capacity fleet management framework~\cite{alonso2017demand}. We improve on both computational efficiency and system performance. The speed-up techniques include search-space pruning and I/O cost reduction for parallelization, reducing the computation time by up to 97.67\%, in experiments on taxi trips in New York City. The proactive rebalancing routes idle vehicles to future demands based on probabilistic estimates from historical demand, increasing the service rate by 4.8\% on average, and decreasing the waiting time and total delay by 5.0\% and 10.7\% on average, respectively. 
\end{abstract}


\IEEEpeerreviewmaketitle

\section{Introduction}
\IEEEPARstart{T}{he} growing popularity of on-demand ridesplitting (trip-sharing) services, such as UberPool and Lyft Line, is moving us towards a more sustainable form of urban mobility\cite{schaller2017unsustainable}. These services allow multiple passengers traveling on similar routes to share a ride. They not only provide passengers with a reliable and low-cost travel mode but also help reduce the congestion and pollution associated with transportation. Higher capacity ridesourcing services such as Via also offer affordable alternatives to public transit services. \par
While these services are currently operated by fleets with human drivers, the operators are aggressively pursuing fleets of autonomous vehicles, with millions of dollars of investments in the autonomous vehicle industry. 
Therefore, demand-responsive autonomous mobility services are gaining research interest and have the potential for reducing operational costs in moving passengers~\cite{shaheen2019shared}. \par
Efficient fleet management methods for optimizing the assignment of vehicles to travel requests is critical for enabling high-quality on-demand ridesourcing services. Most of the literature on Mobility-on-Demand (MoD) services has focused on ride-hailing systems (without shared rides) until recently \cite{zhang2016control}. A recent study showed that 80\% taxi trips in Manhattan, New York City (NYC) could be shared by two riders with only a few minutes travel time increase \cite{santi2014quantifying}, using a optimal trip matching algorithm, which was one of the first attempts to quantify the benefits of pooling ridesourcing services on an urban-scale network. Unfortunately, this matching problem can only be solved optimally in polynomial time when at most 2 requests are shared by a vehicle, and is intractable for high capacity settings. At higher capacities, this problem is commonly known as the dynamic vehicle routing problem (VRP) with time windows and has a large literature that spans many decades~\cite{golden2008vehicle}. However, techniques to solve this problem efficiently at an urban scale have not emerged until recently. In a recent study~\cite{alonso2017demand}, a computationally efficient anytime optimal fleet management framework was proposed to solve this problem at scale. The framework addressed both the problem of assigning vehicles to travel requests and the problem of rebalancing idle vehicles. Traditional methods for fleet management usually formulate the problem as an Integer Linear Program (ILP) and use various heuristics to solve the ILP, but most these methods are either not tractable for large-scale cases or do not provide good results. The framework presented in~\cite{alonso2017demand} decoupled the routing and matching components of the problem; first employing the pairwise shareability graph~\cite{santi2014quantifying} to compute feasible trips and then solving a compacter ILP to find the optimal assignment of vehicles to trips. While the framework was shown to be scaleable for real-time operations in most of the experiments based on taxi trips in NYC, the computation time was found to be too high in cases where the fleet size was large. For example, considering a fleet of 3000 four-seat vehicles, the simulation for travel demand of one day could be as high as about 41 hours using a 24-core computer~\cite{alonso2017demand}. 
\par

In this work, we present a series of techniques to improve the performance of the above-mentioned framework in \cite{alonso2017demand}. The contributions of this work include:

\begin{itemize}[noitemsep,topsep=0pt]
\item Two search space pruning techniques for computing feasible trips. As mentioned above, the method in \cite{alonso2017demand} first uses the shareability graph to find all feasible trips and then solves a compacter ILP to find the optimal assignment between vehicles and trips. The computation is the most intensive in the first step, since it involves solving a large number of small-scale generalized pickup-and-delivery (PDP) problems. These two techniques focus on reducing the computation time of the generalized PDP by pruning the search space based on feasibility constraints.  \par

\item An Input/Output (I/O) reduction technique for parallelization. In~\cite{alonso2017demand}, the process for finding feasible request-vehicle allocations is parallelized. The requests are partitioned and the compatibility checks are done in parallel in separate compute units. However, since the set of candidate vehicles for different requests can overlap with each other, the same vehicle information might be passed to many compute units creating an I/O bottleneck. Therefore, we propose a method for efficiently partitioning requests such that the I/O overhead is minimized. \par

\item A probabilistic proactive rebalancing method based on demand distributions obtained from historical data. This method guides idle vehicles to areas with a high probability of future requests, as determined by the marginal probability of an additional vehicle being utilized at a given location. While other predictive rebalancing methods exist, they are usually based on sampling future demands from historical demand distribution and routing idle vehicles to these regions~\cite{alonso2017predictive, huang2018efficient}. By computing the estimated marginal probability of a request assignment, our proposed method goes beyond this and increases the probability of a rebalancing vehicle being used. \par
\end{itemize}

We numerically evaluate the performance of the proposed techniques using taxi trips in Manhattan, NYC.
\par

The rest of this article is structured as follows. In Section~\ref{sec: formulation}, we present the problem definition and briefly review the fleet management method in \cite{alonso2017demand}. In Section~\ref{sec: techniques}, we introduce speed-up techniques to improve computation performance. Section~\ref{sec: rebalancing} proposes a proactive rebalancing method to improve the MoD system performance. Section~\ref{sec: experiments} consists of numerical experiments on the Manhattan network to show the efficiency of the techniques. Finally, Section~\ref{sec: conclusion} provides closing remarks and discusses possible directions for future research. \par

\section{Preliminaries}\label{sec: formulation}
In this section, we present the problem definition and then briefly introduce the state-of-the-art high capacity MoD service simulation framework in \cite{alonso2017demand}.

\subsection{Problem definition}\label{sec: problem}
We consider a fleet of vehicles with a given capacity to satisfy the MoD service demand. Two main tasks of the MoD fleet management are: i) to assign the vehicles to travel requests; ii) to rebalance idle vehicles to areas where would likely to have travel demand. \par

The set of vehicles is denoted by $\mathcal{V}$. Each vehicle $v_i$ has a corresponding capacity $c_i$. The set of travel requests is denoted by $\mathcal{R}$. A \emph{passenger} is defined as a request that has already been picked up by a vehicle and is on its path to the destination. For each request $r$, the waiting time denoted by $w_r$ is defined as the difference between the pickup time $t_r^p$ and the request time $t_r^r$. For each request $r$ that has been picked up (i.e. passenger), the total travel delay is defined as $\delta_r = t_r^d - t_r^*$, where $t_r^d$ is the dropoff time and $t_r^*$ is the earliest possible time at which the destination could be reached. We assume $t_r^*$ to be the travel time between the origin $o_r$ and the destination $d_r$ by following the shortest path. The goal is to find the optimal assignment that minimizes a cost function $C$ under the following constraints: \par

\begin{itemize}
\item For each request $r$, the waiting time $w_r$ must be below the maximum waiting time $\Omega_r$. \par

\item For each request $r$ that has been picked up, the total delay $\delta_r$ must not exceed the maximum delay $\Delta_r$. \par
\end{itemize}

If a request is not served under the above constraints, it is discarded in the simulation, and a large penalty $c_{ko}$ will occur in the cost function. The cost function $C$ is set to the sum of total delay and the penalty for unassigned requests.\footnote{We use this cost function to be consistent with the framework in \cite{alonso2017demand}. In practice, we can use other cost functions. } The total delay includes both the waiting time for all assigned requests and the in-vehicle delay caused by sharing with other passengers. \par

$$C = \sum\limits_{v\in \mathcal{V}}\sum\limits_{r\in \mathcal{P}_v} \delta_r + \sum\limits_{r\in \mathcal{R}_{ok}} \delta_r + \sum\limits_{r\in \mathcal{R}_{ko}} c_{ko}$$
where $\mathcal{P}_v$ is the set of passengers in vehicle $v$, $\mathcal{R}_{ok}$ is the set of requests that have been assigned vehicles, and $\mathcal{R}_{ko}$ is the set of requests that are not assigned any vehicles. \par

\subsection{Method overview}\label{sec: method}
The framework in \cite{alonso2017demand} decouples the problem by first computing feasible trips based on a pairwise shareability graph \cite{santi2014quantifying} and then finding the optimal trip-vehicle assignment by solving an ILP of reduced dimensionality. The vehicle assignment is processed every 30 seconds. After an assignment round, requests that have not been picked up remain in the request pool until the waiting time constraint is violated. Specifically, each round of the assignment simulation includes the following steps. \par

i) Construct a pairwise request-vehicle graph (RV-graph). RV-graph describes possible pairwise matchings between vehicles and requests. In the graph, two requests $r_i$ and $r_j$ are connected if a virtual vehicle at the origin of either request can serve both requests under the constraints discussed in Section~\ref{sec: problem}; Similarly, a request $r$ is connected to a vehicle $v$ if $v$ can serve $r$ without violating the constraints. \par

ii) Compute the Request-Trip-Vehicle graph (RTV-graph) using the cliques of the RV-graph. A feasible trip is defined as a set of requests that can be served by one vehicle without violating the constraints. RTV-graph consists of all the feasible trips and the vehicles that can serve them. A request $r$ is connected to a trip $T$ if $T$ contains $r$; A trip $T$ is connected to a vehicle $v$ if $v$ can serve $T$ under the system constraints. \par

iii) Solve an ILP to find the optimal assignment from vehicles to feasible trips. The ILP in \cite{alonso2017demand} is as follows: \par

\begin{equation}\label{old_assign}
\Scale[0.865]{
\begin{array}{ll@{}ll}
\text{minimize}  & \displaystyle\sum\limits_{i, j\in X_{TV}} c_{i, j} \cdot x_{i, j} + \sum\limits_{k \in \{1, ..., n\}}c_{ko} \cdot \chi_k &\\
\text{subject to}& \displaystyle\sum\limits_{i \in \mathcal{I}_{V = j}^{T}}   x_{i, j} \leq 1,  &\forall v_j \in \mathcal{V}\\
                 & \sum\limits_{i\in \mathcal{I}_{R = k}^T} \sum\limits_{j\in \mathcal{I}_{T = i}^V} x_{i, j} + \chi_k = 1, &\forall r_k\in \mathcal{R}\\ & x_{i, j} \in \{0,1\}, &\forall i, j\in X_{TV}
\end{array}
}
\end{equation}
where $X_{TV}$ is the set of all feasible assignments between trips and vehicles, $c_{i, j}$ is the cost of vehicle $j$ serving trip $i$. The decision variables are $x_{i, j}$ and $\chi_k$, where $x_{i, j} = 1$ if vehicle $j$ is assigned to trip $i$, and $\chi_k = 1$ if the request $r_k$ is unassigned in this round of assignment simulation. In the constraints, there are three sets: the set of trips that can be served by vehicle $j$ is $\mathcal{I}_{V = j}^{T}$; the set of trips that contains request $r_k$ is $\mathcal{I}_{R = k}^T$; the set of vehicles that can serve trip $i$ is $\mathcal{I}_{T = i}^{V}$. \par

iv) Rebalance idle vehicles. After the assignment, there may be idle vehicles (no trips assigned to the vehicle) and unsatisfied requests (no vehicles assigned). Assuming that more requests will be likely to appear in the neighborhood of the unsatisfied requests, a linear program (LP) is solved to match the idle vehicles to the locations of unsatisfied requests while minimizing the rebalancing cost. \par

\begin{equation}\label{old_rebalance}
\begin{array}{ll@{}ll}
\text{minimize}  & \displaystyle\sum\limits_{v \in \mathcal{V}_{\text{idle}}}\sum\limits_{r \in \mathcal{R}_{ko}} \tau_{v, r} y_{v, r} &\\
\text{subject to}& \displaystyle\sum\limits_{v \in \mathcal{V}_{\text{idle}}}\sum\limits_{r\in \mathcal{R}_{ko}} y_{v, r} = \min(|\mathcal{V}_{\text{idle}}|, |\mathcal{R}_{ko}|)  &\\
                 & 0\leq y_{v, r} \leq 1 \quad \forall y_{v,r}\in \mathcal{Y}&
\end{array}
\end{equation}
where $\mathcal{V}_{\text{idle}}$ is the set of idle vehicles, $\tau_{v, r}$ is the travel time between the vehicle $v$'s location and the origin of the unassigned request $r$, $y_{v, r}$ is the decision variable where $y_{v, r} = 1$ if the vehicle $v$ is assigned the task of traveling to the origin of $r$ and 0 otherwise, and $\mathcal{Y}$ is the set of decision variables. \par

\section{Speed-up techniques to improve on-demand high capacity ridesourcing services} \label{sec: techniques}
In this section, we propose two speed-up techniques for improving the computation time of the framework in \cite{alonso2017demand}. These techniques are exact in the sense that they do not compromise the anytime optimality of the framework.\par

\subsection{Search space pruning}\label{sec: pruning}
When constructing the RV-graph and the RTV-graph in~\cite{alonso2017demand}, an exhaustive search is conducted to check which vehicles can serve each trip. For each trip and vehicle, a generalized pickup-and-delivery-problem (PDP) with time windows and capacity constraints is conducted. For low vehicle capacities (four or less), this is done via enumerating every possible order of pickups and drop-offs and then checking for the violation of any constraints. For larger capacities, an insertion heuristic is used. The best feasible order (if one exists) is picked for each trip-vehicle pair. This is a time-consuming process for trips that include a high number of requests. In addition, the number of trips increases exponentially with the number of requests. Therefore, we propose two techniques to: i) reduce the number of PDP instances that need to be solved; ii) decrease the computation time of each such computation. The following claims are based on the assumption that the travel time in the network is not decreasing\footnote{This is a reasonable assumption in these settings since; i) the performance guarantees are given with respect to expected travel-times based on some travel time estimation model, and ii) if real-time traffic conditions are used, they almost always only get worse.}.\par 

\begin{myCla}\label{cla: feasi}
If a vehicle $v$ is not feasible for trip $T$ at time $t$, it will not be feasible for $T$ at any time $t^\prime (t^\prime > t)$. 
\end{myCla}
\begin{proof}
As the definition reveals, vehicle $v$ is not feasible for trip $T$ at time $t$ represents that any possible route of $v$ serving trip $T$ will violate at least one of the constraints. As the travel time in the network is not decreasing from $t$ to $t^\prime$, any possible route of $v$ serving trip $T$ at $t^\prime$ still violates at least one of the constraints. 
\end{proof}

In \cite{alonso2017demand}, when a vehicle $v$ is not feasible for a trip in the current round of assignment simulation, it is still considered for the trip in the following rounds. To reduce such duplicate computation, we propose the following technique: \par

For any trip $T$, an array $A_{T}$ is used to record the set of feasible vehicles corresponding to trip $T$. At each round of assignment, we incorporate the following procedure.\\
i) If a new request (one that occurred in the last 30 seconds) is assigned to a trip $T$, initialize $A_{T}$ to be the empty set; \\
ii) Check the feasibility of assigning each vehicle $v\in V$ to $T$. If $v$ is feasible for $T$, add $v$ to $A_{T}$. 
\begin{myCla}\label{cla: t1}
Only vehicles in $A_{T}$ need to be considered for trip $T$. If $|A_{T}| = 0$, trip $T$ can be ignored in subsequent rounds. 
\end{myCla}
Claim~\ref{cla: feasi} trivially leads to the correctness of Claim~\ref{cla: t1} as an infeasible trip can not be feasible without a reduction in travel-times. \par

However, as the number of trips increases exponentially with the number of requests, this process can be prohibitively memory-consuming if we consider all trips. Furthermore, the computation time saved by this technique is diminishing as the number of requests in a trip increases, since the set of feasible vehicles for larger trip sizes is much smaller. Therefore, we only use the technique for trips with a single request. Even when considering networks with dynamic travel times, the technique can be used as an efficient heuristic, since the travel times will rarely decrease during a trip. \par

As mentioned previously, for each trip and vehicle, we need to solve a generalized PDP to check the feasibility of the trip-vehicle pairing. In~\cite{alonso2017demand}, this is done by performing the following checks during the construction of each candidate route: i) the waiting time constraint of a request when the vehicle traverses its origin, and ii) the delay constraint of a request when the vehicle reaches the request's destination. We can improve this process with the following modification. Given a vehicle $v$ and a trip $t$, we check the feasibility of $v$ with respect to each waiting time and delay constraint for each request $r \in T$ at each step of the route construction. This technique can decrease the computation time for solving the generalized PDP. \par

The intuition is that we can proactively eliminate infeasible routes by considering potential constraint violations in advance. Inserting a request into a specific route may immediately make other requests (both already inserted and to be inserted) that belong to this trip to be infeasible under this route. Eliminating infeasible routes as early as possible in the search process leads to faster computation. For example, consider a vehicle $v$ assigned to a route $[o_1, o_2... , o_n, ..., d_1, ..., d_n]$ where $o_i$ and $d_i$ represent the origin and destination of request $i$, respectively. The assignment of $o_2$ as the second pickup point might immediately make it impossible to satisfy the pickup time constraint for $o_n$ and the detour constraint for $d_1$. Pruning out infeasible routes as early as possible will reduce computation time if the saving exceeds the overhead of the additional checks. Our experimental results validate that the procedure leads to reduced overall computation time. \par

\subsection{Parallelization I/O bottleneck}
When constructing the RV-graph, each request is checked against to all vehicles in the fleet to determine which vehicles satisfy its constraints. Specifically, for each request-vehicle pair, the framework checks if there exists a route that can serve the request and all the passengers in the vehicle without violating their constraints. This process is done via enumeration and is parallelized across the requests. Each request is assigned to a computing unit in the pool, and information on the set of candidate vehicles is also input to it. Here \textit{candidate vehicles} represents vehicles that can travel directly to the origin of the request and pick up the request within the waiting time constraint. The set of candidate vehicles for different requests may overlap with each other. Therefore, to minimize the I/O requirements of the system, it would be advantageous to allocate requests with similar candidate vehicle sets to the same computing unit. More formally, to minimize the I/O overhead, we propose the following optimization for clustering the requests: \par

\begin{alignat}{3}
  & \text{minimize}   & \quad & \sum\limits_{i \in \mathcal{N}} |\bigcup\limits_{r \mid z_{i, r} = 1} \mathcal{V}_{r}| &          \\
  & \text{subject to} &       & \sum\limits_{i \in \mathcal{N}} z_{i, r} = 1 &  \quad\forall r \in \mathcal{R} \label{const: request_once}\\
  &                   &       &  z_{i, r} \in \{0,1\} & \quad\forall i \in N, r \in \mathcal{R} 
\end{alignat}
where $\mathcal{N}$ is the set of computing units, $\mathcal{V}_r$ is the set of possible vehicles for request $r$, $z_{i, r}$ is the decision variable, where $z_{i, r} = 1$ represents that request $r$ is assigned to computing unit $i$, and 0 otherwise. The objective is to minimize the total number of vehicles that are allocated to each computing unit (collectively). The problem is similar to the weighted set partitioning problem, but has the added complexity of the weight for each set not being known until the assignment $z_{i, r}$ is made. \par

Intuitively, the set of candidate vehicles should be similar for requests that are spatially close to each other, Therefore, instead of solving the above problem directly, we use the K-means clustering algorithm \cite{macqueen1967some} to cluster the requests based on their coordinates, and assign each cluster of requests to different computing units. The number of clusters $k$ is set to the number of computing units in the computer. While this approach leads to good experimental results, more advanced approximation methods such as~\cite{dutta2018hashing} can also be used. \par

\section{Proactive rebalancing of idle vehicles} \label{sec: rebalancing}
In this section, we propose a proactive probabilistic vehicle rebalancing method for improving the service efficiency of the MoD system. Specifically, we incorporate the following changes: i) a new rebalancing formulation, which guarantees that the solution is a one-to-one matching between the idle vehicles and unsatisfied requests; ii) a probabilistic approach for incorporating proactive vehicle rebalancing. \par

\subsubsection{New vehicle rebalancing formulation} \label{sec: new_reba_form}
The rebalancing approach in~\cite{alonso2017demand} routes empty vehicles to locations with unserved demand after each round of trip-vehicle assignment. It is possible for a vehicle to be assigned different rebalancing destinations in consecutive iterations, and potentially circle around inefficiently. Therefore, we first add a constraint to ensure that a vehicle that is already rebalancing will not be re-routed to another destination. Importantly, this does not preclude a rebalancing vehicle from being re-routed to pick up a real demand . \par

In the original formulation (discussed in Section~\ref{sec: method}), multiple vehicles may be assigned to a single unserved request when there are many empty vehicles close to a request. This can lead to vehicles clustering around a neighborhood (with supply exceeding expected demand) and decrease the spatial coverage of the fleet. Therefore, we include an additional constraint to the rebalancing problem to avoid this. \par

\begin{equation}\label{new_rebalance}
\begin{array}{ll@{}ll}
\text{minimize}  & \displaystyle\sum\limits_{v \in \mathcal{V}_{\text{idle}}}\sum\limits_{r \in \mathcal{R}_{ko}} \tau_{v, r} y_{v, r} &\\
\text{subject to}& \displaystyle\sum\limits_{v \in \mathcal{V}_{\text{idle}}}\sum\limits_{r\in \mathcal{R}_{ko}} y_{v, r} = \min(|\mathcal{V}_{\text{idle}}|, |\mathcal{R}_{ko}|)  &\\

&\displaystyle\sum\limits_{v \in \mathcal{V}_{\text{idle}}}   y_{v, r} \leq 1,  \quad \forall r \in \mathcal{R}_{ko}&\\
                 & 0\leq y_{v, r} \leq 1, \quad \forall y_{v,r}\in \mathcal{Y}&
\end{array}
\end{equation}

 This revised formulation ensures that the solution is a one-to-one matching between the vehicles and the requests. However, it also increases the number of constraints by $|\mathcal{R}_{ko}|$. Formulating this problem explicitly (for a solver) can be prohibitively time-consuming when $|\mathcal{R}_{ko}|$ is large. Therefore, in our experiments, we set an upper bound to both the number of vehicles and requests in the formulation, which are denoted by $V^{\text{max}}$ and $R^{\text{max}}$, respectively. Specifically, this is done via the following steps: \par

\textbf{Step 1:} If $|\mathcal{R}_{ko}| > R^{\text{max}}$, we randomly sample $R^{\text{max}}$ requests from $\mathcal{R}_{ko}$ as the set of requests used in the rebalancing procedure; \par

\textbf{Step 2:} If $|\mathcal{V}_{\text{idle}}| > \min\left(V^{\text{max}}, \gamma \cdot \min(|\mathcal{R}_{ko}|, R^{\text{max}})\right)$, we randomly sample $\min\left(V^{\text{max}}, \gamma \cdot \min(|\mathcal{R}_{ko}|, R^{\text{max}})\right)$ vehicles from $\mathcal{V}_{\text{idle}}$ as the set of vehicles used in the rebalance. \par

We apply a constant multiplier $\gamma$, which corresponds to $\gamma$ idle vehicles being considered for rebalancing on average for each unassigned request. The parameters $R^{\text{max}}$, $V^{\text{max}}$ and $\gamma$ control the tradeoff between the computation time and the rebalancing optimality at the moment.\footnote{In the experiments, $\gamma = 3$, $V^{\text{max}} = 300$ and $R^{\text{max}} = 600$. } \par

The rebalancing technique used in \cite{alonso2017demand} assumes that the demand is stationary - i.e. new requests will appear in the same areas where the unserved requests are. However, this assumption does not hold through during many time periods of the day. Therefore, we propose a new probabilistic rebalancing approach for incorporating future demand information more accurately via spatio-temporal demand distributions that are estimated from historical data. \par

Our work is not the first extension of this framework that considers predictive rebalancing strategies. In particular, in~\cite{alonso2017predictive} the authors built a probability distribution over the future demand according to historical data, and then sample virtual requests from this distribution which were added to the request pool. These virtual future requests guided the vehicles to the areas of historical future demand through the trip-vehicle assignment. However, we note from their results that: i) the service rate (number of requests served) was approximately the same as the reactive method without sampled virtual requests. A potential explanation is that the simple demand sampling is not effective in capturing the most important locations to rebalance to; ii) the computation time increased significantly when adding several hundred virtual requests in each round of assignment simulation. Therefore, we introduce a new probabilistic rebalancing method that addresses both of these limitations. \par

Since the demand is typically sparsely distributed across thousands of nodes in an urban-size network, it is hard to predict future demand at the node level for a short time window. As we only need to guide the vehicles to areas where requests are likely to appear, we use a clustering algorithm to spatially cluster the nodes according to their geo-coordinates. We assume that each passenger has a walking range of $\alpha$ miles 
. Then, the number of clusters $k$ is determined by satisfying $\frac{\text{Total area}}{k} \approx 2\pi \alpha^{2}$. Based on historical data, we build the probability distribution $P(n \mid o, \xi)$, which is the probability of $n$ requests appearing in cluster $o$ given a time interval $\xi$. A time interval $\xi$ is defined by a tuple $(t_{s}, t_{e})$ where $t_{s}$ and $t_{e}$ are the start time and the end time of the time interval, respectively. Let $n_o^{\xi}$ be the maximum number of requests appearing in cluster $o$ at time interval $\xi$ according to historical data. We can generate $\sum\limits_{o \in O} n_o^{\xi}$ virtual requests for time interval $\xi$, where $O$ is the set of clusters. For each virtual request $r$, we set the probability of it appearing at time interval $\xi$, which is denoted by $p_r$, as follows. \par

Assume that we are considering cluster $o$ and time interval $\xi$. The virtual requests are denoted by $r_1$, $r_2$, ..., $r_{n_o^{\xi}}$. As the definition reveals, $P(n \mid o, \xi)$ is the joint probability of exactly $n$ requests occurring among all $n_o^{\xi}$ requests. What we want is the marginal probability $p_{r_i}$ for $1 \leq i \leq n_o^{\xi}$. We cannot solve all these marginal probabilities because the correlations between the requests are unknown. To help get a solution, we assume that $r_i$ can only appear when all requests in the set $\{r_j \mid 1 \leq j < i\}$ appear. This assumption implies the order on these requests, i.e., when there are $i$ requests appearing, they would be $r_1$, $r_2$,..., $r_i$. Under this assumption, we derive that $p_{r_i} = \sum\limits_{n = i}^{n_{o}^{\xi}} P(n \mid o, \xi)$. \par

We can use the formulation in Section~\ref{sec: new_reba_form} considering the set of virtual requests with a probability higher than $p_{\text{min}}$\footnote{We let $p_{\text{min}}$ be 0.75 in the numerical experiments. }. The number of virtual requests can be large when the demand is high. To reduce the scale of the optimization problem, we ignore the virtual requests if there exists an idle vehicle within a predefined travel time range\footnote{We let the predefined range to be half of the maximum waiting time in the experiments. }. 

In summary, this method has the following advantages: i) Our method is able to more effectively match the rebalancing vehicles with the most likely locations in which they will be useful, since we explicitly compute the marginal probabilities for the vehicle to be matched with a future request.; ii) Our approach does not increase the scale of any of the optimization problems that need to be solved. \par


\begin{table*}[]
\centering
\caption{Performance comparison. }
\label{tab: alg}
\resizebox{0.88\textwidth}{!}{%
\begin{tabular}{|c|c|c|c|c|c|c|c|c|}
\thickhline
Number of vehicles & Capacity & $\Omega$ (s) & $\Delta$ (s) & Method & Computation time (s) & Service rate & Waiting time (s) & Total delay (s) \\ \thickhline
\multirow{3}{*}{1000} & \multirow{18}{*}{4} & \multirow{9}{*}{120} & \multirow{9}{*}{240} & Original & 4.56 & 0.445 & 73.58 & 132.82 \\ \cline{5-9} 
 &  &  &  & Speed-up & 0.78 & 0.461 & 73.67 & 132.43 \\ \cline{5-9} 
 &  &  &  & Speed-up+Proactive & 1.02 & 0.460 & 72.91 & 131.13 \\ \cline{1-1} \cline{5-9} 
\multirow{3}{*}{2000} &  &  &  & Original & 8.10 & 0.697 & 71.70 & 120.70 \\ \cline{5-9} 
 &  &  &  & Speed-up & 1.38 & 0.768 & 72.11 & 117.51 \\ \cline{5-9} 
 &  &  &  & Speed-up+Proactive & 2.41 & 0.762 & 70.41 & 114.57 \\ \cline{1-1} \cline{5-9} 
\multirow{3}{*}{3000} &  &  &  & Original & 10.37 & 0.785 & 69.37 & 110.18 \\ \cline{5-9} 
 &  &  &  & Speed-up & 1.92 & 0.904 & 67.62 & 95.84 \\ \cline{5-9} 
 &  &  &  & Speed-up+Proactive & 4.28 & 0.913 & 64.66 & 87.81 \\ \cline{1-1} \cline{3-9}\cline{3-9} 
\multirow{3}{*}{1000} &  & \multirow{9}{*}{300} & \multirow{9}{*}{600} & Original & 45.46 & 0.608 & 172.54 & 377.55 \\ \cline{5-9} 
 &  &  &  & Speed-up & 2.92 & 0.609 & 171.28 & 373.78 \\ \cline{5-9} 
 &  &  &  & Speed-up+Proactive & 2.61 & 0.609 & 171.30 & 373.70 \\ \cline{1-1} \cline{5-9} 
\multirow{3}{*}{2000} &  &  &  & Original & 58.26 & 0.950 & 147.79 & 305.12 \\ \cline{5-9} 
 &  &  &  & Speed-up & 5.04 & 0.957 & 146.97 & 301.02 \\ \cline{5-9} 
 &  &  &  & Speed-up+Proactive & 4.79 & 0.956 & 143.78 & 293.73 \\ \cline{1-1} \cline{5-9} 
\multirow{3}{*}{3000} &  &  &  & Original & 60.14 & 0.988 & 119.87 & 200.15 \\ \cline{5-9} 
 &  &  &  & Speed-up & 5.25 & 0.992 & 116.53 & 185.97 \\ \cline{5-9} 
 &  &  &  & Speed-up+Proactive & 6.02 & 0.996 & 95.59 & 131.64 \\ \thickhline
\multirow{3}{*}{1000} & \multirow{18}{*}{10} & \multirow{9}{*}{120} & \multirow{9}{*}{240} & Original & 5.20 & 0.454 & 73.14 & 132.86 \\ \cline{5-9} 
 &  &  &  & Speed-up & 0.76 & 0.467 & 73.55 & 132.42 \\ \cline{5-9} 
 &  &  &  & Speed-up+Proactive & 1.03 & 0.464 & 72.71 & 131.16 \\ \cline{1-1} \cline{5-9} 
\multirow{3}{*}{2000} &  &  &  & Original & 8.85 & 0.701 & 71.61 & 120.67 \\ \cline{5-9} 
 &  &  &  & Speed-up & 1.37 & 0.769 & 72.07 & 117.34 \\ \cline{5-9} 
 &  &  &  & Speed-up+Proactive & 2.43 & 0.763 & 70.53 & 114.82 \\ \cline{1-1} \cline{5-9} 
\multirow{3}{*}{3000} &  &  &  & Original & 11.24 & 0.793 & 69.07 & 109.42 \\ \cline{5-9} 
 &  &  &  & Speed-up & 2.81 & 0.904 & 67.80 & 96.19 \\ \cline{5-9} 
 &  &  &  & Speed-up+Proactive & 4.26 & 0.913 & 64.63 & 88.26 \\ \cline{1-1} \cline{3-9}\cline{3-9} 
\multirow{3}{*}{1000} &  & \multirow{9}{*}{300} & \multirow{9}{*}{600} & Original & 120.75 & 0.661 & 158.26 & 373.29 \\ \cline{5-9} 
 &  &  &  & Speed-up & 2.81 & 0.667 & 161.17 & 372.05 \\ \cline{5-9} 
 &  &  &  & Speed-up+Proactive & 2.63 & 0.659 & 160.32 & 369.98 \\ \cline{1-1} \cline{5-9} 
\multirow{3}{*}{2000} &  &  &  & Original & 80.39 & 0.960 & 140.30 & 297.76 \\ \cline{5-9} 
 &  &  &  & Speed-up & 4.85 & 0.963 & 143.61 & 299.02 \\ \cline{5-9} 
 &  &  &  & Speed-up+Proactive & 4.80 & 0.962 & 139.85 & 289.87 \\ \cline{1-1} \cline{5-9} 
\multirow{3}{*}{3000} &  &  &  & Original & 49.61 & 0.989 & 117.27 & 198.59 \\ \cline{5-9} 
 &  &  &  & Speed-up & 5.25 & 0.991 & 115.47 & 184.88 \\ \cline{5-9} 
 &  &  &  & Speed-up+Proactive & 5.95 & 0.996 & 95.19 & 131.33 \\ \thickhline
\end{tabular}%
}
\end{table*}
\section{Numerical experiments}\label{sec: experiments}
\subsection{Experimental setup}
We use taxi trip data in Manhattan, NYC from 6 am to 12 pm on an arbitrary day (Monday, May 6th, 2013) \cite{donovan2014new} as the travel demand. The network we use is the entire road network of Manhattan (4092 nodes and 9453 edges) \cite{santi2014quantifying,alonso2017demand}. The link travel time is the daily mean travel time, which is computed using the method in \cite{santi2014quantifying}. We evaluate the performance of the techniques by comparing the metrics of interest in varying cases (different fleet sizes, capacities, maximum waiting time, maximum delay) between three methods: i) the framework in \cite{alonso2017demand}, denoted by \textit{original}; ii) the framework with the speed-up techniques\footnote{Note that the rebalancing formulation~\ref{new_rebalance} is also used in the framework with speed-up techniques, denoted by \textit{Speed-up}.}, denoted by \textit{speed-up}; iii) the framework with speed-up techniques and the proactive rebalancing, denoted by \textit{speed-up+proactive}. The system is implemented using Python $3.5$, and all experiments are conducted on a 4 core 3.4GHz computer. The maximum waiting time and the maximum total delay are assumed to be the same for all requests, which are denoted by $\Omega$ and $\Delta$, respectively. \par

\subsection{Experimental results}
We conduct two sets of experiments with a fleet of capacity 4 vehicles and a fleet of capacity 10 vehicles, respectively. During the experiments, we collect the following metrics: the mean computation time (average computation time for a round of assignment simulation), service rate, mean waiting time and mean total delay. These metrics are shown in Table~\ref{tab: alg}. \par

The speed-up techniques decrease the computation time significantly in all cases with no optimality loss. Specifically, the average computation time reduction is as high as 87.46\%, while the mean waiting time and mean in-vehicle delay remain at the same level. The maximum computation time reduction is 97.67\% when operating a fleet of 1000 capacity 10 vehicles with $\Omega = 300$ and $\Delta = 600$. When $\Omega$ and $\Delta$ increases, the computation time of the \textit{original} framework increases significantly, but the computation time of the framework with speed-up techniques remains at the same magnitude. In addition, the service rate is increased by 3.5\% on average, which corresponds to 4068 more passengers being served in 6 hours. These service rate improvements can potentially be attributed to two features of our system: i) the search space pruning techniques reduce the computation time needed for trip feasibility checks, thereby providing additional time for exploring more feasible solutions in the RTV-graph, given a computation time budget\footnote{In \cite{alonso2017demand}, when computing the RTV-graph, a time limit to explore potential trips per vehicle is imposed.}; ii) the rebalancing formulation in Equation~\ref{new_rebalance} is more effective due to the better alignment of vehicles with future requests. 
The maximum gap in service rate (11.9\%) happens when operating a fleet of 3000 capacity 4 vehicles with $\Omega = 120$ and $\Delta = 240$. \par

The framework with speed-up techniques and proactive rebalancing offers an 81.44\% computation time reduction on average compared to the \textit{original} framework. The average computation reduction is slightly lower than the framework with only the speed-up techniques, because the virtual requests can increase the complexity of the ILP (when higher than the number of unassigned requests).
However, the service rate is increased by 4.8\% on average, which is higher than with just the speed-up techniques. In addition, the waiting time and total delay are decreased by 5.0\% and 10.7\%, respectively. The largest gap happens when operating a fleet of 3000 capacity 4 vehicles when $\Omega = 120$ and $\Delta = 300$, where the waiting time and the total delay are reduced by 24.28 seconds and 68.51 seconds on average, respectively. \par

\section{Conclusion}\label{sec: conclusion}
We presented a series of techniques to improve ride-vehicle assignment and fleet rebalancing, within the context of the state-of-the-art Mobility-on-Demand (MoD) service fleet management framework from~\cite{alonso2017demand}. We experimentally show that the proposed speed-up techniques can reduce the computation time by up to 97.67\% for the representative day in NYC. We also propose a proactive probabilistic rebalancing method, which increases the service rate by 4.8\% on average, and decreases the waiting time and total delay by 5\% and 10.7\% on average. We believe that these techniques will enable faster and more accurate system simulation and provide insights for implementing such a framework in industry. Future extensions include: i) designing speed-up techniques that can be employed in networks with dynamic travel times; 
ii) analyzing the sensitivity of the hyperparameters in the rebalancing methods. \par



\bibliographystyle{IEEEtran}
\bibliography{sample}

\begin{thebibliography}{10}
\providecommand{\url}[1]{#1}
\csname url@samestyle\endcsname
\providecommand{\newblock}{\relax}
\providecommand{\bibinfo}[2]{#2}
\providecommand{\BIBentrySTDinterwordspacing}{\spaceskip=0pt\relax}
\providecommand{\BIBentryALTinterwordstretchfactor}{4}
\providecommand{\BIBentryALTinterwordspacing}{\spaceskip=\fontdimen2\font plus
\BIBentryALTinterwordstretchfactor\fontdimen3\font minus
  \fontdimen4\font\relax}
\providecommand{\BIBforeignlanguage}[2]{{%
\expandafter\ifx\csname l@#1\endcsname\relax
\typeout{** WARNING: IEEEtran.bst: No hyphenation pattern has been}%
\typeout{** loaded for the language `#1'. Using the pattern for}%
\typeout{** the default language instead.}%
\else
\language=\csname l@#1\endcsname
\fi
#2}}
\providecommand{\BIBdecl}{\relax}
\BIBdecl

\bibitem{alonso2017demand}
J.~Alonso-Mora, S.~Samaranayake, A.~Wallar, E.~Frazzoli, and D.~Rus,
  ``On-demand high-capacity ride-sharing via dynamic trip-vehicle assignment,''
  \emph{Proceedings of the National Academy of Sciences}, vol. 114, no.~3, pp.
  462--467, 2017.

\bibitem{schaller2017unsustainable}
B.~Schaller, ``Unsustainable? the growth of app-based ride services and
  traffic, travel and the future of new york city,'' \emph{Schaller
  Consulting}, 2017.

\bibitem{shaheen2019shared}
S.~Shaheen and A.~Cohen, ``Shared ride services in north america: definitions,
  impacts, and the future of pooling,'' \emph{Transport reviews}, vol.~39,
  no.~4, pp. 427--442, 2019.

\bibitem{zhang2016control}
R.~Zhang and M.~Pavone, ``Control of robotic mobility-on-demand systems: a
  queueing-theoretical perspective,'' \emph{The International Journal of
  Robotics Research}, vol.~35, no. 1-3, pp. 186--203, 2016.

\bibitem{santi2014quantifying}
P.~Santi, G.~Resta, M.~Szell, S.~Sobolevsky, S.~H. Strogatz, and C.~Ratti,
  ``Quantifying the benefits of vehicle pooling with shareability networks,''
  \emph{Proceedings of the National Academy of Sciences}, vol. 111, no.~37, pp.
  13\,290--13\,294, 2014.

\bibitem{golden2008vehicle}
B.~L. Golden, S.~Raghavan, and E.~A. Wasil, \emph{The vehicle routing problem:
  latest advances and new challenges}.\hskip 1em plus 0.5em minus 0.4em\relax
  Springer Science \& Business Media, 2008, vol.~43.

\bibitem{alonso2017predictive}
J.~Alonso-Mora, A.~Wallar, and D.~Rus, ``Predictive routing for autonomous
  mobility-on-demand systems with ride-sharing,'' in \emph{Intelligent Robots
  and Systems (IROS), 2017 IEEE/RSJ International Conference on}.\hskip 1em
  plus 0.5em minus 0.4em\relax IEEE, 2017, pp. 3583--3590.

\bibitem{huang2018efficient}
X.~Huang and H.~Peng, ``Efficient mobility-on-demand system with
  ride-sharing,'' in \emph{2018 21st International Conference on Intelligent
  Transportation Systems (ITSC)}.\hskip 1em plus 0.5em minus 0.4em\relax IEEE,
  2018, pp. 3633--3638.

\bibitem{macqueen1967some}
J.~MacQueen \emph{et~al.}, ``Some methods for classification and analysis of
  multivariate observations,'' in \emph{Proceedings of the fifth Berkeley
  symposium on mathematical statistics and probability}, vol.~1, no.~14.\hskip
  1em plus 0.5em minus 0.4em\relax Oakland, CA, USA, 1967, pp. 281--297.

\bibitem{dutta2018hashing}
C.~Dutta, ``When hashing met matching: Efficient search for potential matches
  in ride sharing,'' \emph{arXiv preprint arXiv:1809.02680}, 2018.

\bibitem{donovan2014new}
B.~Donovan and D.~Work, ``New york city taxi trip data (2010--2013),'' 2014.

\end{thebibliography}

\end{document}